\documentclass[conference]{IEEEtran}
\IEEEoverridecommandlockouts
\usepackage{cite}
\usepackage{amsmath,amssymb,amsfonts}
\usepackage{algorithmic}
\usepackage{graphicx}
\usepackage{textcomp}
\usepackage{amsthm}
\newtheorem{theorem}{Theorem}
\newtheorem{remark}{Remark}
\newtheorem{example}{Example}

\newtheorem{lemma}{Lemma}
\usepackage{xcolor}
\usepackage{graphicx}
\usepackage{tikz,pgfplots}

\newcommand{\set}[1]{\left\lbrace #1 \right\rbrace}
\newcommand{\dbc}[1]{\left[ #1 \right]}
\newcommand{\brc}[1]{\left(#1 \right)}
\newcommand{\norm}[1]{\left\Vert #1 \right\Vert}
\newcommand{\abs}[1]{\left\vert #1 \right\vert}
\newcommand{\Ex}[1]{\mathcal{E}\left\lbrace #1 \right\rbrace}
\DeclareMathOperator*{\argmin}{argmin}

\newcommand{\trp}{\mathsf{T}}
\newcommand{\her}{\mathsf{H}}
\newcommand{\Diag}{\mathrm{Diag}}
\newcommand{\jj}{\mathrm{j}}
\newcommand{\e}{\mathrm{e}}
\newcommand{\zf}{{\mathrm{zf}}}
\newcommand{\rmg}{{\mathrm{g}}}

\newcommand{\bmm}{\mathbf{m}}
\newcommand{\bh}{\mathbf{h}}
\newcommand{\bg}{\mathbf{g}}
\newcommand{\bo}{\mathbf{o}}
\newcommand{\bn}{\mathbf{n}}

\newcommand{\by}{\mathbf{y}}
\newcommand{\bx}{\mathbf{x}}
\newcommand{\bphi}{\boldsymbol{\phi}}

\newcommand{\mI}{\mathbf{I}}
\newcommand{\mH}{\mathbf{H}}
\newcommand{\mB}{\mathbf{B}}
\newcommand{\mPhi}{\boldsymbol{\Phi}}

\newcommand{\setC}{\mathbb{C}}
\newcommand{\setS}{\mathbb{S}}
\newcommand{\setD}{\mathbb{D}}
\newcommand{\setDb}{\bar{\mathbb{D}}}
\newcommand{\setDzf}{{\mathbb{D}^{{\mathrm{zf}}} }}

\newtheorem{definition}{Definition}
\def\BibTeX{{\rm B\kern-.05em{\sc i\kern-.025em b}\kern-.08em
    T\kern-.1667em\lower.7ex\hbox{E}\kern-.125emX}}
\begin{document}

\title{How to Coordinate Edge Devices for Over-the-Air Federated Learning?}

\author{
	\IEEEauthorblockN{
		Mohammad Ali Sedaghat\IEEEauthorrefmark{1},
		Ali Bereyhi\IEEEauthorrefmark{2},
		Saba Asaad\IEEEauthorrefmark{2}, and
		Ralf R. M\"uller\IEEEauthorrefmark{2},
}
	\IEEEauthorblockA{
		\IEEEauthorrefmark{1}Qualcomm CDMA Technologies GmbH Nuremberg \\
		\IEEEauthorrefmark{2}Friedrich-Alexander Universit\"at Erlangen-Nürnberg (FAU)\\
		\texttt{msedagha@qti.qualcomm.com, \{ali.bereyhi, saba.asaad, ralf.r.mueller\}@fau.de}
}
}
\maketitle

\begin{abstract}
	This work studies the task of device coordination  in wireless networks for over-the-air federated learning (OTA-FL). For conventional metrics of aggregation error, the task is shown to describe the zero-forcing (ZF) and minimum mean squared error (MMSE) schemes and reduces to the NP-hard problem of subset selection. We tackle this problem by studying properties of the optimal scheme. Our analytical results reveal that this scheme is found by searching among the leaves of a tree with favorable monotonic features. Invoking these features, we develop a low-complexity algorithm that approximates the optimal scheme by tracking a dominant path of the tree sequentially. Our numerical investigations show that the proposed algorithm closely tracks the optimal scheme.
	
\end{abstract}


\section{Introduction}
 Federated learning (FL) is a privacy-preserving framework for distributed learning that enables edge devices to address a learning task, e.g., training a global model,  jointly without need for sharing their local data \cite{li2020federated,konevcny2016federated,bonawitz2019towards}. This ensures device privacy and provides a joint scheme for learning that performs considerably better than individual distributed learning while keeping the communication load of the system tractable \cite{kairouz2021advances, aledhari2020federated, ding2022federated}.
 
The promising learning performance in FL comes along with several challenges which are roughly divided into the three categories of \textit{privacy}, \textit{statistical properties of local datasets} and \textit{communication-efficiency}. With respect to privacy, FL has been shown to be vulnerable against model inversion attack \cite{wei2020federated,fredrikson2015model}. This follows from the fact that edge devices in FL share their local models that are often strongly correlated with their local data \cite{so2020byzantine, zhu2021federated}. Statistical properties of local datasets further refers to the sensitivity of FL against data distribution. Initial studies on FL have shown promising performance and fast convergence of this framework for homogeneous local datasets. Nevertheless, obtaining these favorable behaviors have been shown to face some analytical challenges when we deviate from simple statistical models for local datasets \cite{sattler2019robust}. This is in particular crucial, as the exact statistical model of data is rather unknown in practice.

This study focuses on the latter class of challenges in FL, i.e., \textit{communication-efficiency}, that refers to communication-related challenges, when the FL framework is to be implemented in a wireless network. In fact, FL schemes require several rounds of communication between a central node, often called the parameter server (PS), and the edge devices. This is in particular challenging when the scheme runs over remote wireless devices introducing various sources of imperfection to the setting, e.g., rate-limited communication links and noisy assessment of local model parameters  \cite{shi2020communication, duan2019astraea, mills2019communication}. 

\subsection{Over-the-Air Federated Learning}
FL performs distributed learning in three main phases: 
\begin{enumerate}
	\item Starting with an initial global model, edge devices learn a model locally from their datasets. They then share the local models with the PS.
	\item The PS collects the local models and \textit{aggregates} them into a global model.
	\item The global model is sent back to the devices, and the devices repeat local learning via the new global model.
\end{enumerate}
The PS and edge devices keep iterating these three phases for several communication rounds until they converge.

Classical approaches for implementation of FL in a wireless network suggest a separate decoding and aggregation at the PS\cite{amiri2020machine,hellstrom2021over}. This means that the PS decodes first the received local models and then combines the decoded models into a global model via a predefined \textit{aggregation policy}. An alternative approach is proposed in \cite{yang2020federated} which invokes the idea of \textit{analog function computation} \cite{nazer2007computation} to perform model aggregation directly over the air. We refer to this latter approach as over-the-air FL (OTA-FL) in this paper. OTA-FL invokes the signal combination in the multiple access channel to realize the predefined aggregation policy directly over-the-air. In other words, the edge devices in OTA-FL transmit their local models with proper scaling, such the desired model aggregation is calculated from the a linearly-transformed version of the received signal at the PS. In various use-cases, OTA-FL is more efficient than separate decoding and aggregation \cite{xu2021learning}. In fact, OTA-FL does not need to orthogonalize communication resources, i.e., allocate individual bandwidth and time slot to every device, and hence it enjoys both advantages of lower computation and  resource efficiency \cite{cao2021optimized, fan20221, cao2022transmission}.

\subsection{Contributions and Organization}
As OTA-FL relies on the uplink channel for model aggregation, \textit{coordination} plays a key role in its performance. In fact, the PS and devices need to agree on a communication protocol, such that the over-the-air aggregated model is as close as possible to the desired form. Prior studies often consider a basic approach for coordination that is known to be sub-optimal from the information theoretic point of view. In this work, we study the concept of coordination in OTA-FL in great details. In this respect, we first formulate an \textit{optimal} coordination scheme for a given metric of aggregation error mathematically. We then show that using the conventional metric of mean squared error (MSE) for aggregation error, the \textit{optimal} coordination scheme reduces to an MSE minimization problem, while the commonly-used strategy for coordination is simply zero-forcing. Our analyses show that both strategies reduce to the NP-hard problem of subset selection and hence deal with the same level of hardness. We tackle the problem of subset selection in both minimum MSE and zero-forcing schemes by studying the properties of the solution. Our investigations reveal that the solution in both cases can be found by searching over the leaves of a tree with some monotonic properties. In the light of these properties, we propose low-complexity algorithms which sequentially approximate the solution by tracking a dominant path towards leaves. Our simulation results show that the proposed tree-based algorithm can track the optimal performance very closely while enjoying a drastically lower computational complexity.

The remaining of this manuscript is organized as follows: We first formulate the problem concretely in Section \ref{sec1}. Minimum MSE and zero-forcing coordination schemes are then illustrated in Sections \ref{sec2} and \ref{sec3} , respectively. Sections \ref{sec:4} and \ref{ammse_sec} investigates the analytical properties of the zero-forcing and minimum MSE schemes, respectively, and propose a tree-based algorithm for their approximation with low complexity. The validity of the derivations are confirmed through numerical investigation in Section \ref{sec5}. Finally, Section \ref{sec6} concludes the paper.

\subsection{Notation}
Scalars, vectors and matrices are represented with non-bold, bold lower-case and bold upper-case letters, respectively. The transpose and transpose conjugate of $\mH$ is denoted by $\mH^{\trp}$ and $\mH^\her$, respectively. The Euclidean norm of $\bx$ is shown by $\norm{\bx}$. The notation $\Ex{\cdot}$ denotes mathematical expectation and $\mathcal{CN}\brc{\boldsymbol{\mu},\mathbf{C}}$ refers to a multivariate complex Gaussian distribution with mean vector $\boldsymbol{\mu}$ and covariance matrix $\mathbf{C}$. The complex plane is denoted by $\setC$, and the conjugate of $z\in\setC$ is shown by $z^*$. For brevity, $\set{1,\ldots,N}$ is abbreviated by $\dbc{N}$.

\section{Problem Formulation}\label{sec1}
A wireless network with $L$ single-antenna edge devices and a PS is considered. The devices employ an FL scheme, e.g., federated averaging, to address a common learning task, e.g., training a given model, jointly over their distributed local datasets. The PS is equipped with an $N$-element array-antenna. Throughout the analyses, we assume that $N \geq L$. This assumption is however taken for the sake of brevity. The case of $N< L$ is then discussed separately as a remark; see Remark \ref{remark1}. 

At the beginning of each coherence interval, the devices send predefined pilots in their uplink channels. The PS then estimates the channel state information (CSI) based on the received signals. 
For simplicity, we assume that the pilots are mutually orthogonal and that the estimation error is negligible. Hence, the PS knows the CSI perfectly. Throughout the analyses, we assume that the channels experience frequency-flat fading processes. This models transmission over a channel with a narrow frequency band. The derivations are in principle extendable to wide-band transmission in multi-carrier systems with frequency selective fading processes.

\subsection{OTA-FL Setting}
\label{sec:ota-fl}
Following the FL scheme, each device determines its local model parameters from its local dataset. It then shares this local model with the PS through its uplink channel. Let $s_\ell$ be the model parameter of device $\ell$ that is to be shared with the PS in a given time-frequency slot. Device $\ell$ transmits this parameter by applying a channel-dependent scaling coefficient $b_\ell$, i.e., it transmits
\begin{align}
	x_\ell = b_\ell s_\ell .
\end{align}
The local models $\set{s_\ell }$ are assumed to be independent and identically distributed (i.i.d.) with mean zero and unit variance\footnote{Note that in practice $\set{s_\ell }$ are unbiased versions of the model parameters, e.g., quantized gradient values, which for weakly-correlated local datasets fit into this model \cite{lee2020bayesian}.}. The devices are further restricted to set their average power below a maximum power $P$. This means that over a period of time in which the channel is approximately constant, the transmit signal of device $\ell$ satisfies
\begin{align}    
	\Ex{ \abs{x_\ell}^2 } = \abs{b_\ell}^2 \leq P.
\end{align}

The devices communicate over a fading Gaussian MAC. Hence, the PS receives $\by\in\setC^N$ which is given by
\begin{align}
	\by=\sum_{\ell=1}^L \bh_\ell x_\ell + \bn.
\end{align}
Here, $\bh_\ell \in\setC^N $ denotes the complex uplink channel vector of device $\ell$, and $\bn\in\setC^N$ is additive white Gaussian noise with mean zero and variance $\sigma^2$, i.e.,  $  \bn \sim \mathcal{N} \brc{ \boldsymbol{0},\sigma^2 \mI_N }$.

The ultimate goal of the PS is to combine the local models according to a predefined strategy specified by the FL scheme. More precisely, the PS aims to determine 
\begin{align}
	\theta = \sum_{\ell=1}^{L} \phi_\ell s_\ell \label{eq:FL_Strategy}
\end{align}
for some predefined coefficients $\set{\phi_\ell}$, and then share it with the devices through its downlink channels. To this end, the PS invokes the idea of OTA-FL and combines the local models directly over-the-air:  %
%
%
it  estimates the desired global model, i.e., $\theta$, by combining the received signal elements. In other words, it employs a linear receiver $\bmm \in\setC^N$ and determines an estimate of $\theta$ as
\begin{align}
 \hat{\theta} = \bmm^\trp \by &= \sum_{\ell = 1}^L \bmm^\trp \bh_\ell x_\ell + \bmm^\trp \bn, \\
 &= \sum_{\ell = 1}^L \bmm^\trp \bh_\ell b_\ell s_\ell + \bmm^\trp \bn. \label{eq:hat_theta}
\end{align}

It is readily to shown that by restricting the estimation at the PS to be linear, the OTA-FL scheme is optimal: there exists a receiver $\bmm$ and scaling coefficients $b_\ell$ whose corresponding estimate  of global model, i.e., $\hat{\theta}$ in \eqref{eq:hat_theta}, recovers the estimate achieved by combining the optimal linear estimation of local models according to the predefined strategy in \eqref{eq:FL_Strategy} \cite{yang2020federated}.
 
\subsection{Aggregation Error of the OTA-FL Setting}
As communication is carried out through a noisy network, the estimated global model $\hat{\theta}$ contains some error as compared with the desired global model $\theta$. This error is often called \textit{aggregation error} and can be evaluated in various respect. The classical approach in OTA-FL is to quantify the aggregation error in terms of the mean squared error (MSE).

Using the MSE as the metric, the aggregation error in this setting is given by 
\begin{align}    
	\epsilon \brc{ \bmm , \set{b_\ell} } &= \Ex{\abs{\hat{\theta} - \theta}^2 }, \\
	&= \Ex{ \abs{\sum_{\ell = 1}^L \brc{\bmm^\trp \bh_\ell b_\ell - \phi_\ell} s_\ell + \bmm^\trp \bn }^2 }, \\
	&=  \sum_{\ell = 1}^L \abs{\bmm^\trp \bh_\ell b_\ell - \phi_\ell}^2 + \sigma^2 \norm{\bmm}^2. \label{eq:Aggr_Err1}
\end{align}
The error expression can be compactly represented as
\begin{align}    
	\label{eq:Aggr_error}
	\epsilon \brc{ \bmm , \mB } &= \norm{\mB \mH^\trp \bmm - \bphi }^2  + \sigma^2 \norm{\bmm}^2
\end{align}
where $\mB = \Diag\set{b_\ell}$, $\phi = \dbc{\phi_1,\ldots,\phi_L}^\trp$ and $\mH$ is the uplink channel matrix, i.e., 
\begin{align}
	\mH = \dbc{\bh_1,\ldots,\bh_L}.
\end{align}


\subsection{Coordination Schemes for the OTA-FL Setting}
To establish FL in the network, the PS needs to coordinate the devices periodically: after the uplink training phase at the beginning of a coherence interval, the PS specifies the uplink scaling coefficients of the devices, i.e., $\set{b_\ell}$, and the linear receiver $\bmm$. The devices are then informed about their scaling coefficients by a downlink transmission. Afterwards, throughout the coherence interval, the PS and devices carry out their communication according to the OTA-FL scheme described in Section~\ref{sec:ota-fl} using the coefficients and receiver specified in the training phase. It is worth mentioning that the PS only shares the scaling coefficients with the devices and not the complete CSI.


In general, the optimality of a coordination strategy depends on the choice of metric that determines the quality of the FL scheme. A good approach is to consider the aggregation error as the metric and design the coordination strategy, such that the error is minimized\footnote{This is in particular effective, since the MSE of estimated global model directly affects the learning quality of the FL scheme \cite{yang2020federated}.}. In this case, the optimal coordination strategy is formulated as
\begin{align} \label{MMSE_problem}   
   \brc{\bmm^\star, \mB^\star} = \argmin_{\bmm\in \setC^N , \abs{b_\ell}^2\leq P } \epsilon \brc{ \bmm , \mB } .
   \end{align}
This is in fact a minimum MSE (MMSE) problem. We hence refer to this strategy as MMSE coordination. Due to its non-convexity, MMSE coordination is computationally intractable in a generic network. 

An alternative approach for coordination is to use the idea of zero-forcing (ZF) \cite{yang2020federated,bereyhi2022matching}. In this approach, the PS ignores noise in the channel and coordinates the devices, such that the noise-free version of the estimated global model constructs the desired global model. From \eqref{eq:hat_theta}, this means that the PS finds $\bmm$ and $\set{b_\ell}$, such that
\begin{align}\label{ZF_1}
	\bmm^\trp \bh_\ell b_\ell = \phi_\ell.
\end{align}
In this case, the aggregation error reduces to 
\begin{align} 
 \epsilon \brc{ \bmm , \mB } = \sigma^2 \norm{\bmm}^2. \label{eq:ZF_Aggr_Err}
\end{align}
The PS hence solves the following problem 
\begin{align} \label{ZF_problem}   
	\bmm^{\rm zf} = &\argmin_{\bmm \in\setC^N} \norm{\bmm}^2 \\
	&\text{subject to} \;  \abs{\frac{\phi_\ell}{\bmm^\trp \bh_\ell} }^2\leq P \; \text{ for } \; \ell\in\dbc{L}, \nonumber
\end{align}
to find the receiver. The scaling coefficients $\set{b_\ell}$ are~then~determined from \eqref{ZF_1} by replacing $\bmm$ with $\bmm^{\rm zf}$, i.e., 
\begin{align}\label{ZF_2}
	 b_\ell^{\rm zf} = \frac{\phi_\ell}{ \bh_\ell ^\trp \bmm^{\rm zf} }
\end{align}
In the sequel, we refer to this scheme as ZF coordination.


\section{MMSE Coordination for OTA-FL} \label{sec2}
The design problem for MMSE coordination is non-convex and generally challenging to solve. Nevertheless, for~some~particular cases, the solution can be derived tractably. An instance is given in \cite{liu2020over}, where the solution of a mathematically similar problem is derived for an uplink scenario with a single-antenna PS in  closed form, despite its non-convexity. Motivated by this result, we present in the sequel some preliminary analyses. We show that the problem in the general case with a multi-antenna PS is more complicated and unlike the single-antenna case, the optimal MMSE scheme cannot be derived analytically. The analyses however give some important insights on MMSE coordination that pave the way for the proposed algorithms in next sections. We start the analyses by the following lemma:


\begin{lemma}\label{lemma1}
With MMSE coordination in the network, there is at least one device that transmits with the maximum transmit power $P$.
\end{lemma}
\begin{proof}
The proof is given by contradiction: assume that all the devices transmit with power less than $P$ while coordinated via the MMSE scheme in \eqref{MMSE_problem}. Let $\bmm^\star$ and $\set{b_\ell^\star}$ refer to the solutions of the MMSE scheme. This means that $\abs{b_\ell^\star }^2 < P$. We now construct new scaling coefficients as $\hat{b}_\ell  = \kappa b_\ell^\star$, where $\kappa > 1$ is chosen such that 
\begin{align}
	\max_{\ell \in\dbc{L}}    \abs{\hat{b}_\ell}^2 = P.
\end{align}
A new receiver is further constructed from $\bmm^\star$ as $\hat{\bmm} = \kappa^{-1} \bmm^\star$ and satisfies $\norm{\hat{\bmm}} < \norm{\bmm^\star}$.

We now consider the aggregation error derived in \eqref{eq:Aggr_Err1}. Using the new scaling coefficients and receiver, i.e., $\hat{\bmm}$ and $\{ \hat{b}_\ell \}$, the first term in the error expression remains unchanged as compared with the term given for $\bmm^\star$ and $\set{b_\ell^\star}$. The second term however becomes smaller. This means that we have found  $\hat{\bmm}$ and $\{ \hat{b}_\ell \}$ that satisfy
\begin{align}    
	\epsilon ( \hat{\bmm} , \{ \hat{b}_\ell \} ) \leq \epsilon \brc{ \bmm^\star , \{ b_\ell^\star \} }.
\end{align}
This contradicts with the definition of MMSE coordination in \eqref{MMSE_problem}, and hence concludes the proof\footnote{A similar result has been shown in \cite{liu2020over} for a single-antenna PS.}.
%
\end{proof}
Lemma \ref{lemma1} implies that with MMSE coordination, at least one device is transmitting with the maximum power. From this result, we can have a simple conclusion: let the set of devices $\dbc{L}$ be partitioned as $\dbc{L} = \setD \cup \setDb$, where $\setD$ denotes the set of all devices transmitting with the maximum power $P$ and $\setDb$ represents the complement of $\setD$ including all devices whose transmit powers are less than $P$. Lemma~\ref{lemma1} indicates that $\setD$ is always non-empty, i.e., it contains at least one device. In general, there can be more than a single device in $\setD$. 

We next state Lemma~\ref{lemma2} which describes a key property of the devices in $\setDb$:
\begin{lemma}
	\label{lemma2}
	Let $\bmm^\star$ and $\set{b_\ell^\star}$ be given by MMSE coordination. For any $\ell \in \setDb$, i.e., any device transmitting with power less than $P$, $b_\ell^\star$ satisfies
	\begin{align*}
			b_\ell^\star = \frac{\phi_\ell}{ \bh_\ell ^\trp \bmm^\star }. \label{eq:lem_2}
	\end{align*}
\end{lemma}
\begin{proof}
    The proof is given by contradiction following similar steps taken in the proof of \ref{lemma1}. We hence skip the details.
\end{proof}

To understand this result, let us consider the aggregation error in \eqref{eq:Aggr_Err1}. Lemma~\ref{lemma1} indicates that with MMSE coordination, the components in the first term of the aggregation error that correspond to the devices in $\setDb$, i.e., 
\begin{align}
	\abs{\bmm^{\star \trp} \bh_\ell b^\star_\ell - \phi_\ell}^2 = 0.
\end{align}
In other words, for these devices, MMSE coordination is similar to ZF, c.f. \eqref{ZF_2}, and only the devices in $\setD$, i.e., those who transmit with power $P$, contribute to the first term of the aggregation error. 

We next connect the receiver given by MMSE coordination, i.e., $\bmm^\star$, to the devices in $\setD$. 
\begin{lemma}\label{opt_com}
    Let $\bmm^\star$ be given by MMSE coordination. Then, there exists complex scalars $\alpha_\ell$ for $\ell \in \setD$, such that
    \begin{align}
    	\bmm^\star = \sum_{\ell \in\setD } \alpha_\ell \bh_\ell.
    \end{align}
\end{lemma}
\begin{proof}
    The proof follows the results of Lemmas~\ref{lemma1} and \ref{lemma2}. From Lemma~\ref{lemma2}, we know that 
    \begin{align}    
    	\epsilon \brc{ \bmm^\star , \set{b_\ell^\star}  } &=  \sum_{\ell \in \setD} \abs{\bmm^{\star\trp} \bh_\ell b^\star_\ell - \phi_\ell}^2 + \sigma^2 \norm{\bmm^\star}^2. 
    \end{align}
    Lemma~\ref{lemma1} further implies that for $\ell\in \setD$, the scaling coefficient is $b^\star_\ell = \sqrt{P} \e^{\jj \varphi_\ell}$ for some phase $\varphi_\ell$. Hence, we can write
    \begin{align}    
    	\epsilon \brc{ \bmm^\star , \set{b_\ell^\star}  } &=  \sum_{\ell \in \setD} \abs{\bmm^{\star\trp} \bh_\ell \sqrt{P} \e^{ \jj \varphi_\ell} - \phi_\ell }^2 + \sigma^2 \norm{\bmm^\star}^2 \nonumber \\
    	&\leq  \sum_{\ell \in \setD} \abs{\bmm^{\trp} \bh_\ell \sqrt{P} \e^{ \jj \varphi_\ell} - \phi_\ell }^2 + \sigma^2 \norm{\bmm}^2 \label{eq:upp}
    \end{align}
    for any $\bmm \in \setC^N$. Note that $\varphi_\ell = -\vartheta_\ell $ with $\vartheta_\ell$ being the phase of $\bmm^{\star\trp} \bh_\ell$.
    
    From \eqref{eq:upp}, we can conclude that $\bmm^\star$ is the solution to\footnote{Note that we can set $\varphi_\ell = 0$ without loss of generality.}
    \begin{align}
    	\min_{\bmm \in \setC^N}  \sum_{\ell \in \setD} \abs{\bmm^{\trp} \bh_\ell \sqrt{P} - \phi_\ell }^2 + \sigma^2 \norm{\bmm}^2 \label{eq:obj_mse}
    \end{align}
	that describes a regularized ZF problem. This implies that $\bmm^\star$ only linearly depends on $\set{\bh_\ell: \ell\in\setD}$ and components outside the space spanned by these channel vectors only increase the objective function in \eqref{eq:obj_mse}. This concludes the proof.
\end{proof}

Lemma~\ref{opt_com} describes a genie-aided solution for the MMSE scheme: assume that set $\setD$ is known to us. MMSE coordination is then readily found as follows: we first set,
\begin{align}
	\bmm^\star &= \argmin_{\bmm \in \setC^N}  \sum_{\ell \in \setD} \abs{\bmm^{\trp} \bh_\ell \sqrt{P} - \phi_\ell }^2 + \sigma^2 \norm{\bmm}^2 \\
	&= \argmin_{\bmm \in \setC^N}  \norm{ \sqrt{P} \mH_\setD^\trp \bmm  - \bphi_\setD }^2 + \sigma^2 \norm{\bmm}^2\\
	&\stackrel{\dagger}{=} \frac{1}{\sqrt{P}} \mH_\setD^* \brc{ \mH_\setD^\trp \mH_\setD^* + \frac{\sigma^2}{P} \mI_D }^{-1} \bphi_\setD, \label{eq:mmse_final}
\end{align}
where $\mH_\setD \in \setC^{N \times D}$ with $D= \abs{\setD}$ is the reduced form of the channel matrix $\mH$ including only the columns of $\mH$ whose indices are in $\setD$. Similarly, the vector $\bphi_\setD$ is the reduced form of $\bphi$. The identity $\dagger$ further  follows the regularized ZF solution. Having $\bmm^\star$, we then set $b_\ell^\star= \sqrt{P}$ for $\ell \in \setD$ and the remaining scaling coefficients according to \eqref{eq:lem_2} in Lemma~\ref{lemma2}.

The results of Lemmas~\ref{lemma1}-\ref{opt_com} imply that the main challenge in MMSE coordination is to find set $\setD$, i.e., the devices that need to transmit with maximum power. This is in general an integer programming problem. In Section~ \ref{ammse_sec}, we propose a low-complexity algorithm to approximate this set. To state this algorithm, we need first to present some analytical results on ZF coordination. In the light of these results, we present our low-complexity algorithm for MMSE coordination.

\section{ZF Coordination for OTA-FL}
\label{sec3}
ZF coordination sets the first term of the aggregation error to zero while keeping the transmit powers below $P$. This leads to a closed-form scaling coefficients, c.f. \eqref{ZF_2}, and an expression for the receiver in terms of a quadratic programming problem, i.e., \eqref{ZF_problem}. The ZF scheme can be seen as a mismatched version of MMSE coordination in which the PS postulates noise to be zero-variance. In this section, we invoke this interpretation and extend the analytic results of Section~\ref{sec2} to ZF coordination. 
%

We start the analysis by deriving a \textit{genie-aided} form of ZF coordination. To this end, we partition the devices into subsets $\setD^\zf$ and $\setDb^\zf$ with the former denoting those that transmit with maximum power $P$ and the latter being its~complement.~Similar to $\setD$, the set $\setDzf$ cannot be empty.
\begin{lemma}\label{lemma4}
	With ZF coordination in the network, there is at least one device that transmits with the maximum power $P$.
\end{lemma}
\begin{proof}
	The proof is given by contradiction following identical steps as in the proof of Lemma~\ref{lemma1}. 
\end{proof}

We can further develop a similar result as in Lemma~\ref{opt_com} for ZF coordination:  if one knows $\setD^\zf$; then, the receiver $\bmm^{\rm zf}$ is given by the following lemma.
\begin{lemma}\label{lemma_zf_m}
Given $\setDzf$, the receiver $\bmm^{\rm zf}$ in \eqref{ZF_problem} is given by
\begin{eqnarray}\label{lemmazf_eqn}
\bmm^{\rm zf} =\frac{1}{\sqrt{P}} \mH_\setDzf^* \brc{ \mH_\setDzf^\trp \mH_\setDzf^*}^{-1} \bphi_\setDzf,
\end{eqnarray}  
and the aggregation error is
\begin{eqnarray}\label{mse_zf}
    \epsilon\brc{\bmm^{\rm zf}, \mB^{\rm zf} } = \frac{\sigma^2}{P} \bphi_\setDzf^\trp \brc{ \mH_\setDzf^\trp \mH_\setDzf^*}^{-1} \bphi_\setDzf.
\end{eqnarray}
\end{lemma}
\begin{proof}
	The ZF scheme is a mismatched MMSE coordination scheme in which $\sigma^2$ is postulated to be zero. This implies that Lemma~\ref{opt_com} is also valid for the ZF scheme by replacing $\setD$ with $\setDzf$. Consequently, $\bmm^\zf$ is given in terms of $\setDzf$ from \eqref{eq:mmse_final} by replacing $\setD$ with $\setDzf$ and setting $\sigma^2$ to zero. The aggregation error is then derived by substituting $\bmm^\zf$ into \eqref{eq:ZF_Aggr_Err}. 
\end{proof}

Lemma \ref{lemma_zf_m} implies the same fact about ZF coordination: similar to the MMSE scheme, the root problem in ZF coordination is to find those devices that transmit with the maximum power. It is worth mentioning that $\setD$ and $\setDzf$ are in general not identical; nevertheless, the problem of finding these discrete subsets also reduces to an integer programming problem. For the sake of brevity, we refer to this root problem as the \textit{subset selection} problem in the remaining of this paper. 


\section{Subset Selection for ZF Coordination}
\label{sec:4}
The computational complexity of the subset selection problem grows exponentially with the number of devices: there are in total $2^L- 1$ possible choices\footnote{Note that the empty set does not occur according to Lemmas~\ref{lemma1} and \ref{lemma4}.} for the device subset. Thus, the complete search algorithm poses exponential complexity to the system. We hence look for a low-complexity approach to find a good choice for the device subset.

We start the derivations by defining the concept of a \textit{feasible} setting for ZF coordination:
\begin{definition}
	\label{def:1}
The receiver $\bmm$ and coefficients $\set{b_\ell}$ describe a feasible ZF setting, if $\abs{b_\ell}^2 \leq P$ and $\bmm^\trp \bh_\ell b_\ell = \phi_\ell$ for $\ell \in \dbc{L}$.
\end{definition}
A feasible ZF setting only applies ZF on the received signal, i.e., it sets the first term in the right hand side of \eqref{eq:Aggr_Err1} zero, and ignores the optimization in \eqref{ZF_problem} that minimizes the aggregation error. In this respect, we can see the ZF coordination scheme as a feasible ZF setting with minimal aggregation error. It is further easy to show that a feasible ZF setting is uniquely specified by its set of devices transmitting with power $P$: 
\begin{lemma}
	\label{lemFeasibleZF}
Let $\setS$ denote the set of devices which transmit with power $P$ in a feasible ZF setting with receiver $\bmm$ and scaling coefficients $\set{b_\ell}$. Then, $\bmm$ is given by
\begin{align}
	\label{feasibleZF}
	\bmm =\frac{1}{\sqrt{P}} \mH_\setS^* \brc{ \mH_\setS^\trp \mH_\setS^*}^{-1} \bphi_\setS,
\end{align} 
and $b_\ell$ is determined from \eqref{ZF_1} for $\ell\in\dbc{L}$.
\end{lemma}
\begin{proof}
	The proof takes exactly the same steps as in Lemma~\ref{lemma_zf_m}. We hence skip the proof.
\end{proof}

Considering the definition, we now focus on the number of feasible ZF setting for coordination $C^\zf$. Following Lemmas~\ref{lemma1} and \ref{lemma4}, it is readily concluded that \begin{align}
	C^\zf \leq 2^L-1.
\end{align}
It is further easy to show that $C^\zf$ can be considerably smaller than $2^L-1$: for a non-empty subset $\setS \subseteq \dbc{L}$, determine~${\bmm}$ and $b_\ell$ from Lemma~\ref{lemFeasibleZF}. The scaling coefficients are however not guaranteed to satisfy the transmit power constraint. Thus, the described setting is not necessarily feasible for ZF coordination. This observation implies that not all the $2^L-1$ settings determined by applying ZF over the choices of $\setS$ describe a feasible ZF setting, and hence $C^\zf$ can be in general considerably smaller than $2^L-1$.

Considering the above behavior, an algorithmic approach for approximating the ZF scheme is to track down the feasible ZF settings (or at least a subset of them). If there are only a few of them, i.e., $C^\zf$ is small; then, the ZF coordination scheme is readily found by a finite-dimensional search. The following theorem describes the conditions under which the feasible ZF settings and the ZF scheme are explicitly found.


\begin{theorem} \label{th1}
Let $s$ be index of the device whose channel norm is smallest, i.e., 
\begin{align}
	s = \argmin_{\ell \in \dbc{L} } \norm{\bh_\ell},
\end{align}
and define $\rmg_\ell = \bh_s^\her \bh_\ell / \phi_\ell$ for $\ell\in\dbc{L}$. Assume that $\abs{\rmg_s} \leq \abs{\rmg_\ell}$ for all $\ell \in \dbc{L}$ with equality holding for devices in set $\setS\subseteq  \dbc{L}$. Then, $\setDzf = \setS$, i.e., only the device in $\setS$ transmit with the maximum power $P$. 
The ZF coordination scheme in this case is described with
\begin{eqnarray}
\bmm^\zf = \frac{\phi_s \bh_s^*  }{\sqrt{P} \norm{\bh_s}^2 } .
\end{eqnarray}
\end{theorem}
\begin{proof}
	We start the proof by showing that the given setting is feasible for ZF coordination. To this end, let $\bmm = \bmm^\zf$ as given in the theorem and $b_s = \sqrt{P}$. We first note that
	\begin{align}
		\bmm^\trp \bh_s b_s = \phi_s,
	\end{align}
which satisfies the ZF constraint \eqref{ZF_1} for device $s$. For other devices, we apply ZF by setting
	\begin{align}
	 b_\ell &= \frac{ \phi_\ell}{\bmm^\trp \bh_\ell} = \frac{ \phi_\ell \norm{\bh_s}^2 }{\phi_s \bh_s^\her \bh_\ell} \sqrt{P} = \frac{ \rmg_s }{\rmg_\ell} \sqrt{P}.
\end{align}
This is hence concluded that under the constraint $\abs{\rmg_s} \leq \abs{\rmg_\ell}$, this ZF setting is feasible. Moreover, it is concluded that in this case, all devices with $\abs{\rmg_\ell} = \abs{\rmg_s}$ transmit with power $P$.

We now show that the described feasible ZF setting is in fact the ZF scheme with minimal aggregation error. To this end, let $\bmm^\zf$ and $\set{b_\ell^\zf}$ denote the receiver and scaling coefficients of the ZF scheme, respectively. We decompose $\bmm^\zf$ in terms of its projection on $\bh_s^*$ and its component in the null space of $\bh_s^*$, i.e., 
\begin{align}
	\bmm^\zf = c \; \bh_s^* + c_{\perp} \bh_s^{\perp*}
\end{align}
for some constants $c$ and $c_\perp$ and $\bh_s^\perp$ being a vector orthogonal to $\bh_s$, i.e., $\bh_s^\her \bh_s^\perp = 0$. To fulfill the ZF condition for device $s$, we need to have $\bmm^{\zf\trp} \bh_s  =0$. This concludes that 
\begin{eqnarray}
	 c = \frac{\phi_s}{\norm{\bh_s}^2 b_s },
\end{eqnarray}
and noting that $\abs{b_s}^2\leq P$, we should set 
\begin{eqnarray}
	\abs{c} \geq \frac{\phi_s}{\norm{\bh_s}^2 \sqrt{P} }. \label{proof:1}
\end{eqnarray}
Considering \eqref{eq:ZF_Aggr_Err}, we conclude that the aggregation error is 
\begin{align}
	\epsilon\brc{\bmm^\zf, \set{b_\ell^\zf} } = \sigma^2 \brc{ \abs{c}^2 \; \norm{\bh_s}^2 + \abs{c_{\perp}}^2 \norm{\bh_s^{\perp}}^2 }
\end{align}
that is minimized by setting $\abs{c}$ to its lower bound in \eqref{proof:1} and $c_\perp = 0$. This describes the feasible ZF setting with minimal aggregation error. Hence, the proof is concluded.
\end{proof}

Theorem~\ref{th1} describes a sufficient condition under which the ZF scheme is determined closed-form.  This condition is very likely to hold in practical scenarios in which few devices have significantly larger path-losses as compared with the others edge devices in the network. We consider an example of such a scenario throughout the numerical investigations. It is worth mentioning that in a homogeneous setting, i.e., $\phi_\ell = \phi$ for $\ell\in\dbc{L}$, with a single-antenna PS, the sufficient condition in Theorem \ref{th1} always hold. This is obvious, since in this case only the device with weakest channel should transmit with the maximum power for ZF coordination\footnote{Note that for MMSE coordination, this is in general not true, as also shown for a mathematically-similar problem in \cite{liu2020over}.}.

\subsection{ZF Coordination without Closed-Form Solution}
\label{sec:ZF-noClose}
If the sufficient condition in Theorem~\ref{th1} does not hold, ZF coordination is found by searching over all feasible settings. This search procedure is exponentially hard. We hence develop a tractable algorithm which approximates the solution by polynomial complexity. In this respect, the following lemma, can be further useful to limit our search
\begin{lemma}
\label{lemm_ZF_2}
If the condition in Theorem~\ref{th1} does not hold; then, there are at least two devices which transmit with the maximum power $P$ in ZF coordination.
\end{lemma} 
\begin{proof}
	The proof is readily given by contradiction. We hence skip it here.
\end{proof}

Considering Lemma~\ref{lemm_ZF_2} and the earlier results on ZF coordination, we can have the following statements about the optimal ZF scheme, when it is not given by Theorem~\ref{th1}:
\begin{itemize}
	\item It is a feasible ZF setting.
	\item  There are at least two devices with transmit power $P$.
\end{itemize}
We further note that when all the devices transmit with power $P$, a feasible ZF setting is always guaranteed to be determined: let $b_\ell =\sqrt{P}$ for $\ell\in\dbc{L}$ and $\bmm$ to be the one given in Lemma~\ref{lemma_zf_m} for $\setDzf = \dbc{L}$. These findings leads to this heuristic conclusion that a good approximate for the optimal ZF scheme is given by finding a feasible ZF setting with minimal number of devices transmitting with power $P$. This heuristic conclusion leads us to a tree-based algorithm which is presented in the sequel.

\begin{remark}\label{remark1}
As mentioned in Section~\ref{sec1}, throughout the paper we assume $N\geq L$. This however does not restrict the scope of analyses. In fact, the derivations are further valid for settings with $N<L$ when we limit the subset selection problem to be solved for subsets with $\abs{\setS} \leq N$. 
\end{remark}


\subsection{A Tree-Based Search Scheme}
\label{sec:treeZF}
We now present a tree-based algorithm to find the optimal ZF scheme. Although its core idea comes from our heuristic conclusions, it is shown that with full complexity this algorithm is exact. Before we present the algorithm, we state the following lemma which serves as the analytical foundation of the tree-based algorithm: 


\begin{lemma}\label{lemma_two_set}
	Let $\brc{\bmm_1, \set{b_{1\ell}}}$ and $\brc{\bmm_2, \set{b_{2\ell}}}$ describe~two~feasible ZF setting for coordination. Denote the set of devices transmitting with power $P$ in these two settings with $\setS_1$ and $\setS_2$, respectively. If $\setS_1 \subseteq \setS_2$; then, 
	\begin{align}
		\epsilon \brc{\bmm_1, \set{b_{1\ell}}} \leq \epsilon \brc{\bmm_2, \set{b_{2\ell}}}.
	\end{align}
\end{lemma}
\begin{proof}
The proof is readily concluded by the fact that the feasible region of the error minimization problem for subset $\setS_1$ contains the feasible region of the optimization problem defined with subset $\setS_2$. This implies that the error obtained by the former minimization is smaller than the one in the latter. This concludes the proof.
\end{proof}

The above result implies that while searching for the optimal ZF coordination, those feasible ZF settings whose corresponding subsets include other feasible subsets can be further ignored\footnote{From Lemma~\ref{lemm_ZF_2}, we~further~conclude that this search is restricted to $\abs{S} \geq 2$.}. We hence can graphically visualize the search for the optimal feasible setting via a \textit{tree}: consider a graph whose nodes denote those subsets of $\dbc{L}$ that describe feasible ZF settings. The node corresponding to $\setS \subseteq \dbc{L}$ is connected to its \textit{children} which are the subsets of $\setS$ describing feasible ZF settings. Any set $\setS\neq \dbc{L}$ is further connected to its \textit{parent} which is the feasible ZF set $\tilde{\setS}$ containing $\setS$, i.e., $\setS \subset \tilde{\setS}$. If there are multiple choices for the parent of $\setS$; then, the set with largest cardinality is selected\footnote{If there are multiple nodes with largest cardinality, the we choose the parent among them at random.}. This way we built a tree whose root is $\dbc{L}$ and whose cardinality of the sets corresponding to the nodes decreases as we move towards the leaves. 

%

\begin{example}\label{exp1}
Consider a network with $L=4$ edge devices and $N=5$ antennas at the PS. Let $\phi_\ell = 0.25$ for all devices and $P=1$. The channel matrix is further as follows
\begin{eqnarray}
	\mH = 
     \begin{bmatrix}
   0.30&0.46&0.39&0.19\\
-0.55&0.32&-0.52&0.04\\
0.32&-0.14&-0.48&-0.11\\
0.72&0.13&-0.37&0.18\\
0.21&-0.36&-1.32&-0.23\\
   \end{bmatrix}.
   \end{eqnarray}
The ZF tree for this network is shown in Fig. \ref{fig_tree}. As observed, the root is $\set{1,2,3,4}$.   
\begin{figure}
\begin{center}
\includegraphics[width=3in,angle=0]{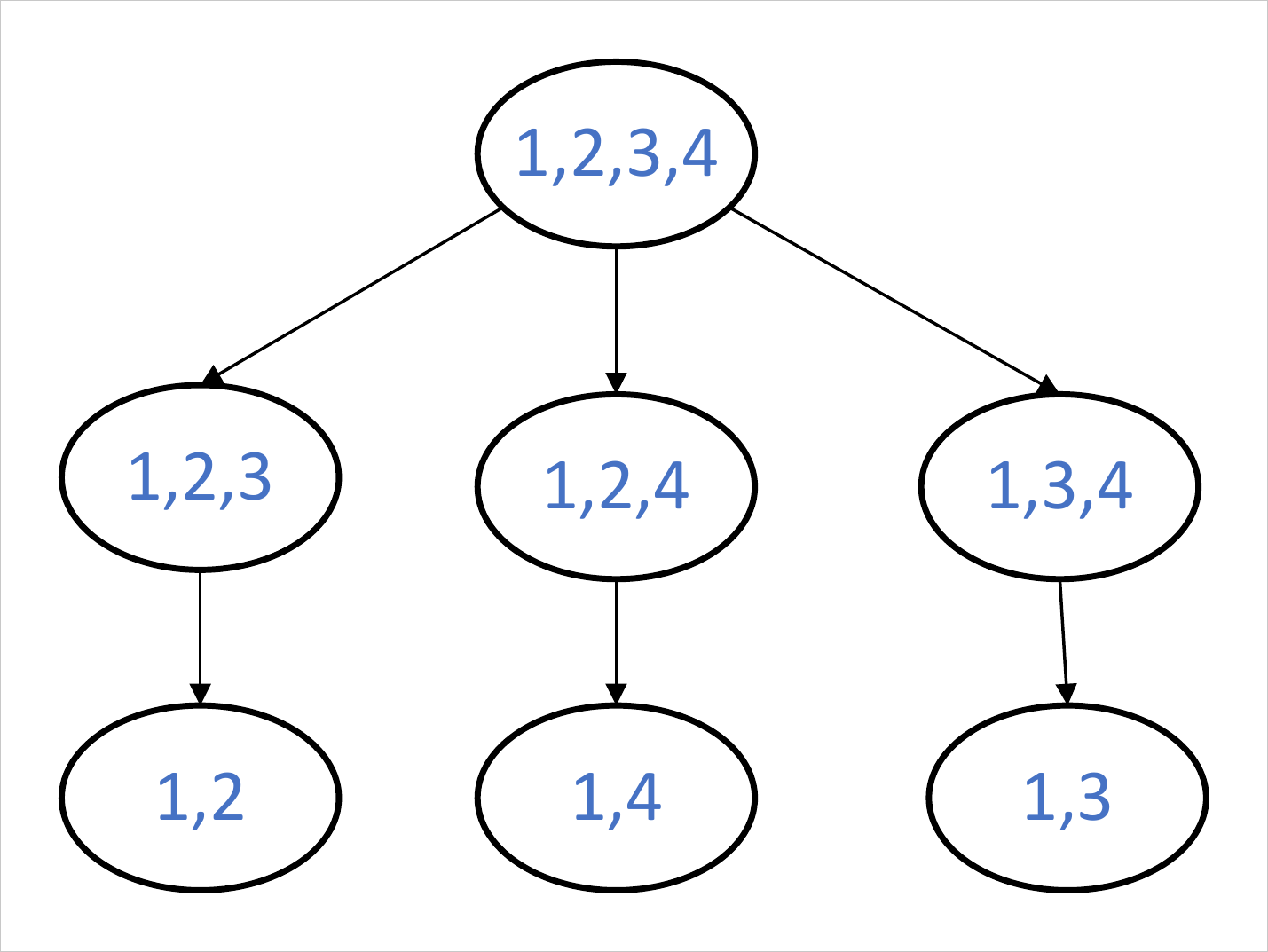}
\caption{Tree representation of feasible ZF settings for Example \ref{exp1}.}
\end{center}
\label{fig_tree}
\end{figure}
\end{example}

Lemma \ref{lemma_two_set} implies that the optimal ZF scheme with minimal aggregation error is among the leaves of the tree representing ZF feasible settings. Therefore, the problem of finding the optimal ZF scheme reduces to a search for the leaf with minimal error, i.e., the leaf with minimal norm of the receiver. It is worth mentioning that this search is still exponentially complex, since in general construction of the tree requires $2^L-1$ feasibility checks. Its visual representation however describes a tractable approximation approach: starting from the root, we move towards the leaves in a step-wise fashion. In each step, we go from a parent to the child whose aggregation error is minimum among all the children. The search is finally over as we arrive at a leaf. This leaf is taken as the approximate of the solution.

\subsection{Approximate ZF Coordination}
The proposed tree-based search lead to an approximation of ZF coordination which we refer to as approximate ZF (AZF). The details on this algorithm are as follows:
\begin{enumerate}
	\item We initiate the search at the root of the tree by setting $\setS^{(0)} = \dbc{L}$. The initial receiver $\bmm^{(0)}$ and scaling coefficients are further found via Lemma~\ref{lemFeasibleZF}. 
	\item At step $i$, we consider all subsets of $\setS^{(i)}$ that differ with $\setS^{(i)}$ in only one element. For each subset, we determine the receiver from $\bmm^{(i)}$ and the corresponding channel vectors via rank-one update of matrix inverse; see \cite{hager1989updating}. We then calculate the scaling coefficients and check if they satisfy the transmit power constraint. If the power constraint is satisfied, we collect the subset as a feasible ZF subset of $\setS^{(i)}$.
	\item Among the feasible ZF subsets of $\setS^{(i)}$, we set $\setS^{(i+1)}$ to be the one whose receiver has minimal norm. 
\end{enumerate}
The algorithm stops at step $i_{\rm T}$, where there is no feasible ZF subsets for the set $\setS^{(i_{\rm T})}$. 

The AZF scheme needs to search among $L\brc{L+1}/2$ subsets in the worst-case scenario, and hence is computationally tractable. It can be further extended by considering multiple dominant paths of the tree. Such an extension can potentially lead to a more accurate approximation of the ZF scheme~at~the expense of higher computational complexity.

\section{Subset Selection for MMSE Coordination} \label{ammse_sec}
We now get back to the MMSE scheme and utilize the framework developed for ZF coordination to design a tractable algorithm for approximating the MMSE scheme. We start the derivations by defining the less-intuitive concept of a \textit{feasible MMSE setting}. Before stating the definition, let us look back to the key properties of the MMSE scheme: from Lemmas~\ref{lemma1} and \ref{lemma2}, we know that the scaling coefficients in this scheme are either $\sqrt{P}$ or satisfy the ZF constraint. We further note that given the set of devices transmitting with power $P$, the receiver can be found via \eqref{eq:mmse_final}. These features are dual to those satisfied by the ZF scheme. We hence invoke Definition~\ref{def:1} and define a feasible setting for MMSE coordination as follows:
\begin{definition}
	\label{def:2}
	Consider the subset of devices $\setS\subseteq \dbc{L}$ whose cardinality is $\abs{\setS} = S$. Let $\bmm$ be
	\begin{align}
		\bmm = \frac{1}{\sqrt{P}} \mH_\setS^* \brc{ \mH_\setS^\trp \mH_\setS^* + \frac{\sigma^2}{P} \mI_S }^{-1} \bphi_\setS.
	\end{align}
	Let $\abs{b_\ell} = \sqrt{P}$ for $\ell\in\setS$ and determine $b_\ell$ for $\ell\notin\setS$ from the equation $\bmm^\trp \bh_\ell b_\ell = \phi_\ell$. The set $\setS$ is said to describe a feasible MMSE setting for coordination, if 
%
%
 $\abs{b_\ell}^2 < P$  for $\ell\notin\setS$. 
\end{definition}

Lemmas~\ref{lemma1} and \ref{lemma2} along with \eqref{eq:mmse_final} imply that the subset of devices transmitting with power $P$ in the MMSE scheme,~i.e., $\setD$, describes a feasible MMSE setting. In other words, similar to ZF coordination, this definition relaxes the MMSE scheme by dropping the optimality constraint. We next follow the same steps as for ZF coordination to show that the MMSE scheme is given by searching over the tree of feasible settings for the smallest subset.

%
%

\subsection{Properties of Feasible MMSE Settings}
Similar to Theorem \ref{th1}, a sufficient condition can be derived for MMSE coordination, under which the set $\setD$, and thus the MMSE scheme, is determined in closed-form. 
\begin{theorem} \label{mmse_1}
	Let $s$ be index of the device whose channel norm is smallest, i.e., 
	\begin{align}
		s = \argmin_{\ell \in \dbc{L} } \norm{\bh_\ell}.
	\end{align}
	 Define the vector $\bg$ as
	 \begin{align}
	 	\bg = \mPhi^{-1} \brc{ \mH^\her \bh_s + \frac{\sigma^2}{P} \bo_s },
	 \end{align}
	 where $\mPhi = \Diag\set{\phi_\ell}$ and $\bo_s \in \set{0,1}^L$ is a vector with a single non-zero entry at index $s$. Denote the entry $\ell$ of $\bg$ with $\rmg_\ell$ and assume that $\abs{\rmg_s} \leq \abs{\rmg_\ell}$ for all $\ell \in \dbc{L}$ with equality holding for devices in set $\setS\subseteq  \dbc{L}$. Then, $\setD = \setS$, i.e., only the device in $\setS$ transmit with the maximum power $P$. The MMSE coordination scheme in this case is described with
	\begin{eqnarray}
		\bmm^\star = \frac{ \sqrt{P} \phi_s \bh_s^*  }{P \norm{\bh_s}^2 + \sigma^2 }.
	\end{eqnarray}
%
\end{theorem}
\begin{proof}
	The proof follows similar steps as those taken for the proof of Theorem~\ref{th1}. We hence skip the details.
\end{proof}
%

Theorem~\ref{th1} describes the sufficient condition for tractable calculation of the MMSE scheme. For networks which do not satisfy this condition, one can show that similar findings, as those given in Sections~\ref{sec:ZF-noClose} and \ref{sec:treeZF} for ZF coordination, are reported as well in this case. In particular, we can show that feasible MMSE settings reduce in aggregation error as their corresponding subset of devices with maximum transmit power shrinks.

\begin{lemma}\label{lemma_two_set_MMSE}
	Let $\brc{\bmm_1, \set{b_{1\ell}}}$ and $\brc{\bmm_2, \set{b_{2\ell}}}$ describe~two~feasible MMSE settings for coordination whose corresponding subsets are $\setS_1$ and $\setS_2$, respectively. If $\setS_1 \subseteq \setS_2$; then, 
	\begin{align}
		\epsilon \brc{\bmm_1, \set{b_{1\ell}}} \leq \epsilon \brc{\bmm_2, \set{b_{2\ell}}}.
	\end{align}
\end{lemma}
\begin{proof}
	The proof is similar to the proof of Lemma \ref{lemma_two_set}.
\end{proof}

\subsection{Approximate MMSE Coordination}
Lemma~\ref{lemma_two_set_MMSE} indicates that the MMSE scheme can be found by searching for the MMSE feasible setting with smallest subset. We hence follow the same approach as in Section~\ref{sec:treeZF} and build the tree of feasible MMSE schemes. Due to its similarity, we skip the definition of the tree and refer the reader to Section~\ref{sec:treeZF}. We now approximately find the MMSE scheme by moving step-wise from the root towards the leaves while in each step we choose the child whose aggregation error is minimum among all the children connected to the parent.

The proposed tree-based search algorithm for MMSE coordination can be summarized as follows:
\begin{enumerate}
	\item We initiate the search at the root of the tree by setting $\setS^{(0)} = \dbc{L}$. The initial receiver $\bmm^{(0)}$ and scaling coefficients are further found from Definition~\ref{def:2}. 
	\item At step $i$, we consider all subsets of $\setS^{(i)}$ that~differ~with $\setS^{(i)}$ in only one element. For each subset, we determine the receiver from $\bmm^{(i)}$ and the channel vectors via rank-one update. We then calculate the scaling coefficients and check if they describe a feasible MMSE setting. 
	\item Among the feasible MMSE subsets of $\setS^{(i)}$, we set $\setS^{(i+1)}$ to be the one whose aggregation error is minimum. 
\end{enumerate}
The algorithm stops, when no feasible MMSE subsets is found. We refer to this algorithm in the sequel as approximate MMSE (AMMSE) scheme.

%

\begin{figure}
	\begin{center}
%
%
\definecolor{mycolor1}{rgb}{0.85000,0.32500,0.09800}%
\definecolor{mycolor2}{rgb}{0.92900,0.69400,0.12500}%
\definecolor{mycolor3}{rgb}{0.49400,0.18400,0.55600}%
\definecolor{mycolor4}{rgb}{0.46600,0.67400,0.18800}%
\begin{tikzpicture}

\begin{axis}[%
width=2.81in,
height=2.1in,
at={(1.351in,0.869in)},
scale only axis,
xmin=-12,
xmax=22,
xlabel style={font=\color{white!15!black}},
xlabel={$\log \mathrm{SNR}$ in [dB]},
ymin=-27,
ymax=12,
ylabel style={font=\color{white!15!black}},
ylabel={$\log \epsilon \brc{\bmm,\mB}$ in [dB]},
axis background/.style={fill=white},
legend style={legend cell align=left, align=left, draw=white!15!black}
]
\addplot [color=mycolor1, line width=2.0pt, mark size=3.5pt, mark=asterisk, mark options={solid, mycolor1}]
  table[row sep=crcr]{%
20	-23.9931241278023\\
18	-21.9879613988322\\
16	-20.0438681669918\\
14	-18.0499795001649\\
12	-16.0216275639901\\
10	-14.003454706998\\
8	-11.9882544135485\\
6	-10.0092694600872\\
4	-7.97627754734932\\
2	-5.99564292072611\\
-0	-3.95476775052131\\
-2	-1.96113133906773\\
-4	0.0171855368771017\\
-6	2.00678487821729\\
-8	4.03199905168215\\
-10	6.05863163690335\\
};
\addlegendentry{Optimal ZF}

\addplot [color=mycolor2, line width=2.0pt, mark size=3.5pt, mark=triangle, mark options={solid, mycolor2}]
  table[row sep=crcr]{%
20	-24.0706786176228\\
18	-22.1149633206137\\
16	-20.244274094597\\
14	-18.3660999886533\\
12	-16.4837471936465\\
10	-14.715067886015\\
8	-13.0783037809331\\
6	-11.5997987936475\\
4	-10.3047756155089\\
2	-9.26897590643573\\
-0	-8.43202092202597\\
-2	-7.79154967965042\\
-4	-7.28216056743879\\
-6	-6.89780362262103\\
-8	-6.61447546545352\\
-10	-6.4128777555208\\
};
\addlegendentry{Optimal MMSE}

\addplot [color=mycolor3, line width=2.0pt, mark size=3.5pt, mark=o, mark options={solid, mycolor3}]
  table[row sep=crcr]{%
20	-23.9312271439696\\
18	-21.9416276915835\\
16	-19.9927112046312\\
14	-17.9951070509897\\
12	-15.9701724946189\\
10	-13.9499305874427\\
8	-11.936653578406\\
6	-9.96160707854001\\
4	-7.92677478523322\\
2	-5.93914428573982\\
-0	-3.89053132042861\\
-2	-1.90449279555624\\
-4	0.0708757866030982\\
-6	2.06433583033648\\
-8	4.08050110885737\\
-10	6.10141707425475\\
};
\addlegendentry{AZF}

\addplot [color=mycolor4, line width=2.0pt, mark size=3.5pt, mark=square, mark options={solid, mycolor4}]
  table[row sep=crcr]{%
20	-24.0175728800555\\
18	-22.0712087122397\\
16	-20.1980014867848\\
14	-18.3321602531394\\
12	-16.4610331300024\\
10	-14.6996451395405\\
8	-13.0725779074801\\
6	-11.5969289768732\\
4	-10.3042791877624\\
2	-9.26897590643573\\
-0	-8.43202092202597\\
-2	-7.79154967965042\\
-4	-7.28216056743879\\
-6	-6.89780362262103\\
-8	-6.61447546545352\\
-10	-6.4128777555208\\
};
\addlegendentry{AMMSE}

\end{axis}
\end{tikzpicture}%
	\end{center}
	\caption{Aggregation error against SNR for the ZF and MMSE schemes with $N=8$ PS antennas and $L=4$ devices. The tree-based reduced searches lie on the optimal schemes.}
	\label{figgg1}
\end{figure}
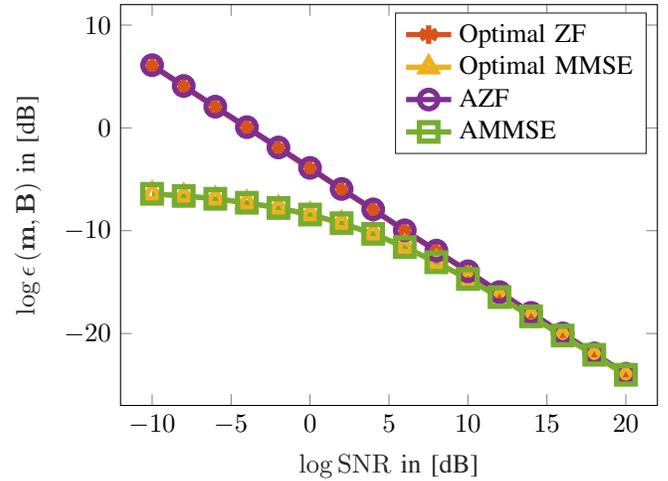

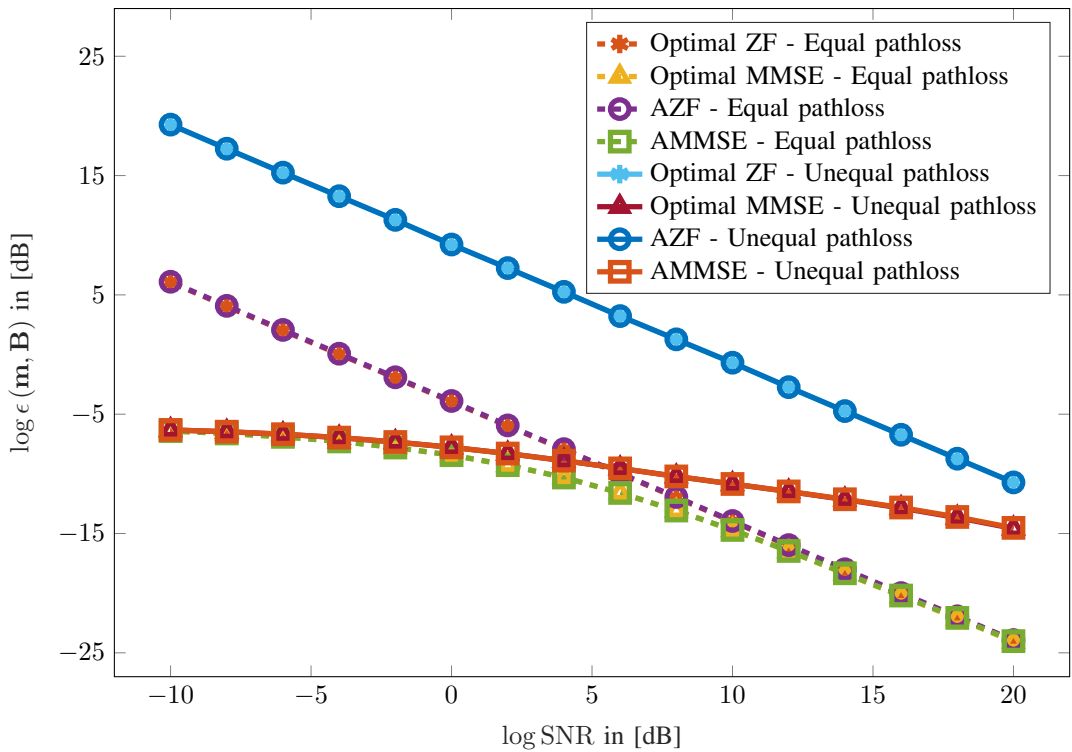
\begin{figure*}
	\begin{center}
%
%
\definecolor{mycolor1}{rgb}{0.85000,0.32500,0.09800}%
\definecolor{mycolor2}{rgb}{0.92900,0.69400,0.12500}%
\definecolor{mycolor3}{rgb}{0.49400,0.18400,0.55600}%
\definecolor{mycolor4}{rgb}{0.46600,0.67400,0.18800}%
\definecolor{mycolor5}{rgb}{0.30100,0.74500,0.93300}%
\definecolor{mycolor6}{rgb}{0.63500,0.07800,0.18400}%
\definecolor{mycolor7}{rgb}{0.00000,0.44700,0.74100}%
\begin{tikzpicture}

\begin{axis}[%
width=5in,
height=3.5in,
at={(2.167in,1.042in)},
scale only axis,
xmin=-12,
xmax=22,
xlabel style={font=\color{white!15!black}},
xlabel={$\log \mathrm{SNR}$ in [dB]},
ymin=-27,
ymax=29,
xtick={-10,-5,0,5,10,15,20},
ytick={-25,-15,-5,5,15,25},
ylabel style={font=\color{white!15!black}},
ylabel={$\log \epsilon \brc{\bmm,\mB}$ in [dB]},
axis background/.style={fill=white},
legend style={legend cell align=left, align=left, draw=white!15!black}
]
\addplot [color=mycolor1, dashed, line width=2.0pt, mark size=3.5pt, mark=asterisk, mark options={solid, mycolor1}]
  table[row sep=crcr]{%
20	-23.9931241278023\\
18	-21.9879613988322\\
16	-20.0438681669918\\
14	-18.0499795001649\\
12	-16.0216275639901\\
10	-14.003454706998\\
8	-11.9882544135485\\
6	-10.0092694600872\\
4	-7.97627754734932\\
2	-5.99564292072611\\
-0	-3.95476775052131\\
-2	-1.96113133906773\\
-4	0.0171855368771017\\
-6	2.00678487821729\\
-8	4.03199905168215\\
-10	6.05863163690335\\
};
\addlegendentry{Optimal ZF - Equal pathloss}

\addplot [color=mycolor2, dashed, line width=2.0pt, mark size=3.5pt, mark=triangle, mark options={solid, mycolor2}]
  table[row sep=crcr]{%
20	-24.0706786176228\\
18	-22.1149633206137\\
16	-20.244274094597\\
14	-18.3660999886533\\
12	-16.4837471936465\\
10	-14.715067886015\\
8	-13.0783037809331\\
6	-11.5997987936475\\
4	-10.3047756155089\\
2	-9.26897590643573\\
-0	-8.43202092202597\\
-2	-7.79154967965042\\
-4	-7.28216056743879\\
-6	-6.89780362262103\\
-8	-6.61447546545352\\
-10	-6.4128777555208\\
};
\addlegendentry{Optimal MMSE - Equal pathloss}

\addplot [color=mycolor3, dashed, line width=2.0pt, mark size=3.5pt, mark=o, mark options={solid, mycolor3}]
  table[row sep=crcr]{%
20	-23.9312271439696\\
18	-21.9416276915835\\
16	-19.9927112046312\\
14	-17.9951070509897\\
12	-15.9701724946189\\
10	-13.9499305874427\\
8	-11.936653578406\\
6	-9.96160707854001\\
4	-7.92677478523322\\
2	-5.93914428573982\\
-0	-3.89053132042861\\
-2	-1.90449279555624\\
-4	0.0708757866030982\\
-6	2.06433583033648\\
-8	4.08050110885737\\
-10	6.10141707425475\\
};
\addlegendentry{AZF - Equal pathloss}

\addplot [color=mycolor4, dashed, line width=2.0pt, mark size=3.5pt, mark=square, mark options={solid, mycolor4}]
  table[row sep=crcr]{%
20	-24.0175728800555\\
18	-22.0712087122397\\
16	-20.1980014867848\\
14	-18.3321602531394\\
12	-16.4610331300024\\
10	-14.6996451395405\\
8	-13.0725779074801\\
6	-11.5969289768732\\
4	-10.3042791877624\\
2	-9.26897590643573\\
-0	-8.43202092202597\\
-2	-7.79154967965042\\
-4	-7.28216056743879\\
-6	-6.89780362262103\\
-8	-6.61447546545352\\
-10	-6.4128777555208\\
};
\addlegendentry{AMMSE - Equal pathloss}

\addplot [color=mycolor5, line width=2.0pt, mark size=3.5pt, mark=asterisk, mark options={solid, mycolor5}]
  table[row sep=crcr]{%
20	-10.7315324423607\\
18	-8.74433164257054\\
16	-6.73802709713781\\
14	-4.75503707507622\\
12	-2.74867446833088\\
10	-0.695258590668927\\
8	1.25741467293774\\
6	3.20997971649789\\
4	5.2404625187235\\
2	7.23387164530381\\
-0	9.20630977590565\\
-2	11.2727314500301\\
-4	13.2538030114207\\
-6	15.2323310253881\\
-8	17.2366446419024\\
-10	19.2575604937508\\
};
\addlegendentry{Optimal ZF - Unequal pathloss}

\addplot [color=mycolor6, line width=2.0pt, mark size=3.5pt, mark=triangle, mark options={solid, mycolor6}]
  table[row sep=crcr]{%
20	-14.5806667357947\\
18	-13.6730619488315\\
16	-12.865439525942\\
14	-12.1537254931733\\
12	-11.4811974058551\\
10	-10.8508201460695\\
8	-10.2143815330727\\
6	-9.54038391853358\\
4	-8.89315578217555\\
2	-8.29640098133033\\
-0	-7.77703766592712\\
-2	-7.32962312517209\\
-4	-6.96074025259916\\
-6	-6.67618649388323\\
-8	-6.46696139045628\\
-10	-6.31744613502728\\
};
\addlegendentry{Optimal MMSE -  Unequal pathloss}

\addplot [color=mycolor7, line width=2.0pt, mark size=3.5pt, mark=o, mark options={solid, mycolor7}]
  table[row sep=crcr]{%
20	-10.7158595909089\\
18	-8.72869917263308\\
16	-6.72277234550613\\
14	-4.73780738279442\\
12	-2.73210883726194\\
10	-0.681805689751401\\
8	1.27356563740013\\
6	3.22311483448676\\
4	5.25587326578506\\
2	7.24770645634299\\
-0	9.22148900713039\\
-2	11.2882462622806\\
-4	13.27170829426\\
-6	15.2468450172697\\
-8	17.2510079969727\\
-10	19.2716695009444\\
};
\addlegendentry{AZF -  Unequal pathloss}

\addplot [color=mycolor1, line width=2.0pt, mark size=3.5pt, mark=square, mark options={solid, mycolor1}]
  table[row sep=crcr]{%
20	-14.5266696587791\\
18	-13.6124789136602\\
16	-12.820232814909\\
14	-12.1290479934142\\
12	-11.4718483920012\\
10	-10.8471544887781\\
8	-10.2142845159864\\
6	-9.54038391853358\\
4	-8.89315578217555\\
2	-8.29640098133033\\
-0	-7.77703766592712\\
-2	-7.32962312517209\\
-4	-6.96074025259916\\
-6	-6.67618649388323\\
-8	-6.46696139045628\\
-10	-6.31744613502728\\
};
\addlegendentry{AMMSE -  Unequal pathloss}

\end{axis}
\end{tikzpicture}%
	\end{center}
	\caption{Aggregation error against SNR for the ZF and MMSE schemes with $N=8$ PS antennas and $L=4$ devices. ZF shows less robustness against asymmetric device distribution.}
	\label{fig2}
\end{figure*}

\section{Numerical results}\label{sec5}
In the section, we examine the efficiency of the proposed algorithms through numerical experiments. To this end, we consider a simple network whose uplink channels experience a standard i.i.d. Rayleigh fading process with fixed path-loss. The channel coefficients are hence written as $h_{\ell,n}=t_\ell f_{\ell,n}$, where $f_{\ell,n}\sim\mathcal{CN}(0,1/N)$ models  small-scale fading, and $t_\ell$ captures the path-loss between the PS and device $\ell$. We further assume an FL setting in which $\phi_\ell=1/L$ for $\ell\in\dbc{L}$. 

In this network, we coordinate the devices using the proposed scheme, i.e., AZF and AMMSE. For sake of comparison, we further evaluate the optimal ZF and MMSE schemes by performing a complete search over the leaves of their corresponding trees. To keep the simulations for the optimum case tractable, we limit the number of devices in the network. 

\subsection{Coordination Performance}
Fig. \ref{figgg1} shows the aggregation error against the signal-to-noise ratio (SNR) in the network, defined as %
	$\mathrm{SNR} = {P}/{\sigma^2}$,
for $L=4$ and $N=8$. For this figure, the path-losses are set to $t_\ell=1$ for $\ell\in\dbc{L}$. As the figure shows, for a small number of devices, the approximate schemes, i.e., AZF and AMMSE, perform optimally.

We next investigate the impact of network asymmetry on both the coordination approaches in Fig.~ \ref{fig2}, where we plot the same figure for two different scenarios: the first scenario is as the one considered in Fig.~\ref{figgg1}, and the other considers the setting in which the path-loss of one of the devices is $10$ dB below the others. As shown in the figure, ZF coordination is more sensitive to the non-uniformity of the path-losses as compared to the MMSE scheme. This is in particular severe at low SNRs which is often the case in current wireless networks.  

\begin{figure}
	\begin{center}
%
%
\definecolor{mycolor1}{rgb}{0.00000,0.44700,0.74100}%
\definecolor{mycolor2}{rgb}{0.85000,0.32500,0.09800}%
\definecolor{mycolor3}{rgb}{0.92900,0.69400,0.12500}%
\definecolor{mycolor4}{rgb}{0.49400,0.18400,0.55600}%
\begin{tikzpicture}

\begin{axis}[%
	width=2.81in,
	height=2.1in,
	at={(1.351in,0.869in)},
scale only axis,
xmin=.5,
xmax=8.5,
xlabel style={font=\color{white!15!black}},
xlabel={load $\xi = N/L$},
ymin=-17,
ymax=-10.5,
xtick={1,2,3,4,5,6,7,8},
ytick={-16,-14,-12,-10},
ylabel style={font=\color{white!15!black}},
ylabel={$\log \epsilon \brc{\bmm,\mB}$ in [dB]},
axis background/.style={fill=white},
legend style={legend cell align=left, align=left, draw=white!15!black}
]
\addplot [color=mycolor1, line width=2.0pt, mark size=3.5pt, mark=asterisk, mark options={solid, mycolor1}]
  table[row sep=crcr]{%
1	-12.2954107335894\\
1.5	-13.4423641988013\\
2	-13.9785575967036\\
2.5	-14.3068737136724\\
3	-14.4839244954256\\
3.5	-14.6211802291753\\
4	-14.7706748480536\\
4.5	-14.858507019649\\
5	-14.9773883885027\\
5.5	-15.0465694963053\\
6	-15.0983280386775\\
6.5	-15.1610464508169\\
7	-15.2284005669544\\
7.5	-15.2600586023595\\
8	-15.2966392299706\\
};
\addlegendentry{Optimal ZF}

\addplot [color=mycolor2, line width=2.0pt, mark size=3.5pt, mark=triangle, mark options={solid, mycolor2}]
  table[row sep=crcr]{%
1	-13.7223187290314\\
1.5	-14.3527919589203\\
2	-14.720758177828\\
2.5	-14.9414600799634\\
3	-15.1187849348591\\
3.5	-15.247109119161\\
4	-15.3759664966696\\
4.5	-15.4516308516677\\
5	-15.5511030401408\\
5.5	-15.6179566321007\\
6	-15.6676356557069\\
6.5	-15.7229165089763\\
7	-15.7846883504168\\
7.5	-15.8106978765082\\
8	-15.8393580263874\\
};
\addlegendentry{Optimal MMSE}

\addplot [color=mycolor3, line width=2.0pt, mark size=3.5pt, mark=o, mark options={solid, mycolor3}]
  table[row sep=crcr]{%
1	-11.8369004433231\\
1.5	-13.3111044415476\\
2	-13.9281575893255\\
2.5	-14.2841000752185\\
3	-14.4754312744799\\
3.5	-14.6171960809579\\
4	-14.7654962303308\\
4.5	-14.858507019649\\
5	-14.9757766087866\\
5.5	-15.0461660216613\\
6	-15.0981646548287\\
6.5	-15.1610464508169\\
7	-15.2278294125721\\
7.5	-15.2600586023595\\
8	-15.2966392299706\\
};
\addlegendentry{AZF}

\addplot [color=mycolor4, line width=2.0pt, mark size=3.5pt, mark=square, mark options={solid, mycolor4}]
  table[row sep=crcr]{%
1	-13.5126467026067\\
1.5	-14.295658568233\\
2	-14.7015242773355\\
2.5	-14.9319361457643\\
3	-15.1162787549576\\
3.5	-15.2462740954258\\
4	-15.3759664966696\\
4.5	-15.4516308516677\\
5	-15.5511030401408\\
5.5	-15.6179566321007\\
6	-15.6676356557069\\
6.5	-15.7229165089763\\
7	-15.7846883504168\\
7.5	-15.8106978765082\\
8	-15.8393580263874\\
};
\addlegendentry{AMMSE}

\end{axis}
\end{tikzpicture}%
	\end{center}
	\caption{Aggregation error against load for the ZF and MMSE schemes with $\log \mathrm{SNR} = 10$ dB and $L=4$ devices.}
	\label{fig3}
\end{figure}
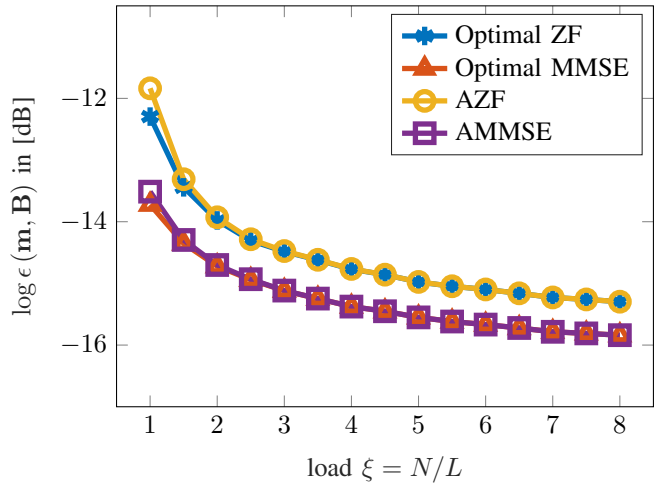

As the next experiment, we plot the aggregation error against the load of the system that is defined as $\xi = N/L$ in Fig. \ref{fig3}. For numerical tractability of the optimal schemes, we set $L=4$. The SNR is further set to $\log\mathrm{SNR} = 12$ dB, and the path-loss is set to $t_\ell = 1$ for $\ell\in\dbc{L}$. As observed, the tree-base approximates closely track the optimal schemes. Furthermore, the gap between the optimal and approximated schemes shrinks as the load grows. 

\subsection{Complexity of the Proposed Schemes}
As a metric of complexity, we define the \textit{check time} to be the number of subsets being checked in the tree-based search of the algorithm. Fig. \ref{fig4} and Fig. \ref{fig5} show the average check time for AZF or AMMSE algorithms against the SNR and load, respectively. Here, $L=4$ and $N=8$, and the check time is averaged over multiple channel realizations. The devices are assumed to be uniformly distributed in the network, i.e., the path-loss is assumed to the same for all the devices. As the figures depict, the complexity of the AZF scheme only varies against the load, while the complexity of the AMMSE scheme scales with both the load and channel quality. The figures further demonstrate that the check time for AMMSE is always less than AZF. This is a favorable behavior, since AMMSE leads to less aggregation error, as well. Fig.~\ref{fig5} further implies that the average check time of both the algorithms reduces by growth of $N$ and converges to $L+1$. This follows the fact that at high system loads, the channel vectors become statistically orthogonal, and hence all the devices transmit with the maximum power. As the result, both the algorithms stop at the root, after checking the feasibility of their $L$ children.

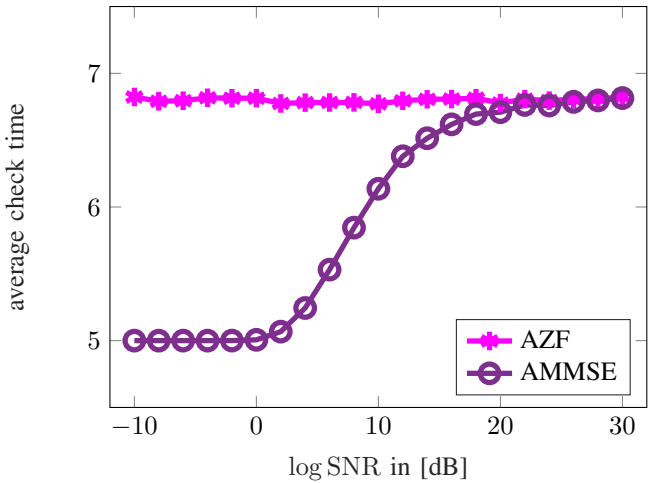
\begin{figure}
	\begin{center}
%
%
\definecolor{mycolor1}{rgb}{1.00000,0.00000,1.00000}%
\definecolor{mycolor2}{rgb}{0.49400,0.18400,0.55600}%
\begin{tikzpicture}

\begin{axis}[%
	width=2.81in,
height=2.1in,
at={(1.351in,0.869in)},
scale only axis,
xmin=-12,
xmax=32,
xlabel style={font=\color{white!15!black}},
xlabel={$\log \mathrm{SNR}$ in [dB]},
ymin=4.5,
ymax=7.5,
xtick={-10,0,10,20,30},
ytick={5,6,7},
ylabel style={font=\color{white!15!black}},
ylabel={average check time},
axis background/.style={fill=white},
title style={font=\bfseries},
legend style={at={(0.97,0.03)}, anchor=south east, legend cell align=left, align=left, draw=white!15!black}
]
\addplot [color=mycolor1, line width=2.0pt, mark size=3.5pt, mark=asterisk, mark options={solid, mycolor1}]
  table[row sep=crcr]{%
30	6.8232\\
28	6.81\\
26	6.8073\\
24	6.7961\\
22	6.807\\
20	6.7776\\
18	6.8142\\
16	6.8078\\
14	6.8057\\
12	6.7945\\
10	6.7719\\
8	6.7826\\
6	6.782\\
4	6.7809\\
2	6.7747\\
-0	6.814\\
-2	6.8139\\
-4	6.8173\\
-6	6.7953\\
-8	6.791\\
-10	6.8231\\
};
\addlegendentry{AZF}

\addplot [color=mycolor2, line width=2.0pt, mark size=3.5pt, mark=o, mark options={solid, mycolor2}]
  table[row sep=crcr]{%
30	6.815\\
28	6.7974\\
26	6.7867\\
24	6.7591\\
22	6.7635\\
20	6.7094\\
18	6.6952\\
16	6.6182\\
14	6.5145\\
12	6.3795\\
10	6.1364\\
8	5.8455\\
6	5.5313\\
4	5.2442\\
2	5.0672\\
-0	5.006\\
-2	5.0003\\
-4	5\\
-6	5\\
-8	5\\
-10	5\\
};
\addlegendentry{AMMSE}

\end{axis}
\end{tikzpicture}%
	\end{center}
	\caption{Average check time against SNR for the ZF and MMSE schemes with $N=8$ PS antennas and $L=4$ devices.}
	\label{fig4}
\end{figure}

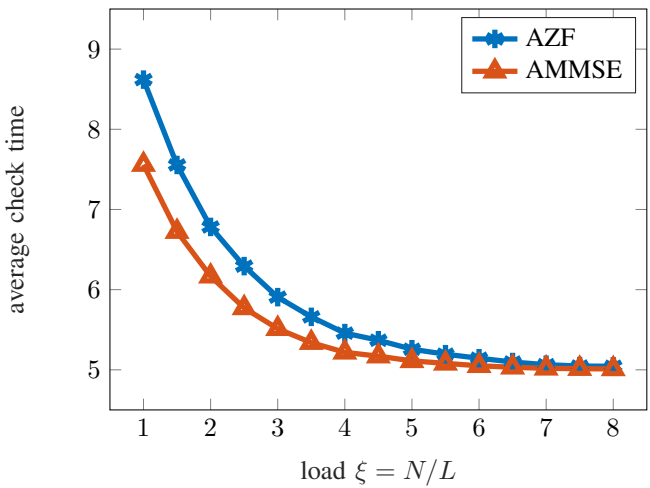
\begin{figure}
	\begin{center}
%
%
\definecolor{mycolor1}{rgb}{0.00000,0.44700,0.74100}%
\definecolor{mycolor2}{rgb}{0.85000,0.32500,0.09800}%
\begin{tikzpicture}

\begin{axis}[%
	width=2.81in,
height=2.1in,
at={(2.167in,1.042in)},
scale only axis,
xmin=.5,
xmax=8.5,
xlabel style={font=\color{white!15!black}},
xlabel={load $\xi = N/L$},
ymin=4.5,
ymax=9.5,
xtick={1,2,3,4,5,6,7,8},
ytick={5,6,7,8,9},
ylabel style={font=\color{white!15!black}},
ylabel={average check time},
axis background/.style={fill=white},
legend style={legend cell align=left, align=left, draw=white!15!black}
]
\addplot [color=mycolor1, line width=2.0pt, mark size=3.5pt, mark=asterisk, mark options={solid, mycolor1}]
  table[row sep=crcr]{%
1	8.6194\\
1.5	7.5552\\
2	6.782\\
2.5	6.2934\\
3	5.9066\\
3.5	5.6607\\
4	5.458\\
4.5	5.3692\\
5	5.2571\\
5.5	5.195\\
6	5.1446\\
6.5	5.0973\\
7	5.0639\\
7.5	5.0471\\
8	5.0446\\
};
\addlegendentry{AZF}

\addplot [color=mycolor2, line width=2.0pt, mark size=3.5pt, mark=triangle, mark options={solid, mycolor2}]
  table[row sep=crcr]{%
1	7.5596\\
1.5	6.7234\\
2	6.1671\\
2.5	5.7719\\
3	5.5123\\
3.5	5.3388\\
4	5.2183\\
4.5	5.1703\\
5	5.1119\\
5.5	5.0804\\
6	5.0507\\
6.5	5.0318\\
7	5.0186\\
7.5	5.0144\\
8	5.0111\\
};
\addlegendentry{AMMSE}

\end{axis}
\end{tikzpicture}%
	\end{center}
	\caption{Average check time against load for the ZF and MMSE schemes with $\log \mathrm{SNR} = 10$ dB and $L=4$ devices.}
	\label{fig5}
\end{figure}

As the last experiment, we investigate the impacts of a non-uniform distribution of the device in the network on tree-based search complexity. Fig. \ref{fig6} shows the average check time for the same setting as in Fig.~\ref{fig5}, when one of the devices is located, such that its path-loss is $10$ dB below the path-losses of the other devices. As the figure shows, this asymmetry increases even further the gap between the complexity of AZF and AMMSE leading to this conclusion that AMMSE proposes a more robust behavior as compared with the AZF.

\begin{figure}
	\begin{center}
%
%
\definecolor{mycolor1}{rgb}{1.00000,0.00000,1.00000}%
\definecolor{mycolor2}{rgb}{0.49400,0.18400,0.55600}%
\begin{tikzpicture}

\begin{axis}[%
	width=2.81in,
height=2.1in,
at={(2.167in,1.042in)},
scale only axis,
xmin=.5,
xmax=8.5,
xlabel style={font=\color{white!15!black}},
xlabel={load $\xi = N/L$},
ymin=4.5,
ymax=11.5,
xtick={1,2,3,4,5,6,7,8},
ytick={5,7,9,11},
ylabel style={font=\color{white!15!black}},
ylabel={average check time},
axis background/.style={fill=white},
legend style={legend cell align=left, align=left, draw=white!15!black}
]
\addplot [color=mycolor1, line width=2.0pt, mark size=3.5pt, mark=asterisk, mark options={solid, mycolor1}]
  table[row sep=crcr]{%
1	10.3864\\
1.5	10.2699\\
2	10.1514\\
2.5	10.0438\\
3	9.9313\\
3.5	9.8439\\
4	9.7325\\
4.5	9.6342\\
5	9.5686\\
5.5	9.5289\\
6	9.4005\\
6.5	9.3476\\
7	9.2786\\
7.5	9.2201\\
8	9.1755\\
};
\addlegendentry{AZF}

\addplot [color=mycolor2, line width=2.0pt, mark size=3.5pt, mark=triangle, mark options={solid, mycolor2}]
  table[row sep=crcr]{%
1	6.6982\\
1.5	6.063\\
2	5.6757\\
2.5	5.4405\\
3	5.294\\
3.5	5.1892\\
4	5.12\\
4.5	5.0747\\
5	5.0555\\
5.5	5.0369\\
6	5.024\\
6.5	5.015\\
7	5.0105\\
7.5	5.0051\\
8	5.0054\\
};
\addlegendentry{AMMSE}

\end{axis}
\end{tikzpicture}%
	\end{center}
	\caption{Average check time against load for the ZF and MMSE schemes with $\log \mathrm{SNR} = 10$ dB and $L=4$ devices in a network with asymmetric device distribution.}
	\label{fig6}
\end{figure}
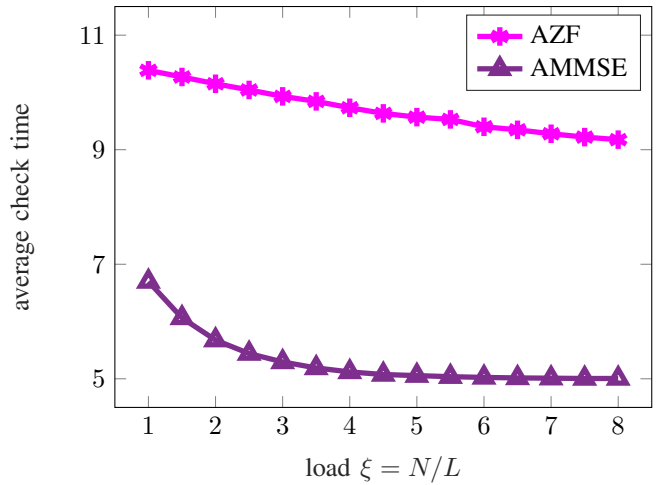

\section{Conclusions}\label{sec6}
Low-complexity algorithms have been proposed for MMSE and ZF coordination in OTA-FL. The algorithms find efficient approximation of the optimal MMSE and ZF schemes using a tree-based search algorithm which solves the equivalent subset selection problem. Our numerical investigations depict that the proposed algorithms track closely the optimal schemes. 

The analytical results of this study indicate that both the MMSE and ZF coordination schemes deal with the same level of hardness, while the former leads to a lower aggregation error. The numerical investigations further reveal that the tree-based approximation of MMSE coordination results in less computational complexity and more robustness as compared to the approximated ZF scheme. Considering the favorable behavior of the MMSE scheme with respect to performance and complexity, one concludes that MMSE is a better approach for device coordination in OTA-FL. 

\bibliographystyle{IEEEtran}
\bibliography{CoordinationOTAFL} 

\end{document}